\documentclass[a4paper,UKenglish,cleveref, autoref]{lipics-v2019}
\title{Sparse Hopsets in Congested Clique
} %TODO Please add

\titlerunning{}%optional, please use if title is longer than one line

\author{Yasamin Nazari}{Johns Hopkins University, Baltimore, MD, United States}{ynazari@jhu.edu}{}{}%TODO mandatory, please use full name; only 1 author per \author macro; first two parameters are mandatory, other parameters can be empty. Please provide at least the name of the affiliation and the country. The full address is optional

\authorrunning{Y. Nazari}%TODO mandatory. First: Use abbreviated first/middle names. Second (only in severe cases): Use first author plus 'et al.'

\Copyright{Yasamin Nazari}%TODO mandatory, please use full first names. LIPIcs license is "CC-BY";  http://creativecommons.org/licenses/by/3.0/

\ccsdesc[500]{Theory of computation~Distributed algorithms}
\keywords{Hopsets, Congested Clique, Shortest Paths, Massively Parallel Computation}%TODO mandatory; please add comma-separated list of keywords

\category{}%optional, e.g. invited paper

\relatedversion{}%optional, e.g. full version hosted on arXiv, HAL, or other respository/website
\supplement{}%optional, e.g. related research data, source code, ... hosted on a repository like zenodo, figshare, GitHub, ...

\funding{This work is supported in part by NSF awards CCF-1464239 and CCF-1909111.}%optional, to capture a funding statement, which applies to all authors. Please enter author specific funding statements as fifth argument of the \author macro.

\acknowledgements{The author would like to thank Michael Dinitz and Goran Zuzic for helpful discussions. }%optional

\nolinenumbers %uncomment to disable line numbering

\hideLIPIcs  %uncomment to remove references to LIPIcs series (logo, DOI, ...), e.g. when preparing a pre-final version to be uploaded to arXiv or another public repository

\EventEditors{Pascal Felber, Roy Friedman, Seth Gilbert, and Avery Miller}
\EventNoEds{4}
\EventLongTitle{23rd International Conference on Principles of Distributed Systems (OPODIS 2019)}
\EventShortTitle{OPODIS 2019}
\EventAcronym{OPODIS}
\EventYear{2019}
\EventDate{December 17--19, 2019}
\EventLocation{Neuch\^{a}tel, Switzerland}
\EventLogo{}
\SeriesVolume{153}
\ArticleNo{26}

\usepackage[ruled,vlined,linesnumbered]{algorithm2e}
\usepackage{algorithmic}

\begin{document}
\maketitle
\begin{abstract}
We give the first Congested Clique algorithm that computes a sparse hopset with polylogarithmic hopbound in polylogarithmic time. Given a graph $G=(V,E)$, a $(\beta,\epsilon)$-hopset $H$ with "hopbound" $\beta$, is a set of edges added to $G$ such that for any pair of nodes $u$ and $v$ in $G$ there is a path with at most $\beta$ hops in $G \cup H$ with length within $(1+\epsilon)$ of the shortest path between $u$ and $v$ in $G$.
 Our hopsets are significantly sparser than the recent construction of Censor-Hillel et al.~\cite{censor2019}, that constructs a hopset of size $\tilde{O}(n^{3/2})$\footnote{The $\tilde{O}(f(n))$ notation is used to hide $O(\log (f(n)))$ factors.}, but with a smaller polylogarithmic hopbound. On the other hand, the previously known construction of sparse hopsets with polylogarithmic hopbound in the Congested Clique model, proposed by Elkin and Neiman \cite{elkin2017, elkin2016, elkin2019RNC}, all require polynomial rounds. 
 
  One tool that we use is an efficient algorithm that constructs an $\ell$-limited neighborhood cover, that may be of independent interest. 
 Finally, as a side result, we also give a hopset construction in a \textit{variant} of the low-memory Massively Parallel Computation model, with improved running time over existing algorithms.
\end{abstract}

\section{Introduction}
In the Congested Clique model of distributed computing, we are given a graph with $n$ nodes, where all nodes can send a (possibly different) message with $O(\log(n))$-bits to \textit{every} other node in the graph in each round. In the context of distributed graph algorithms, the input graph is a subset of the communication graph. In addition to theoretical interest in this model, it has also recently gained a lot of attention, due to its connections to practical distributed and big data platforms such as MapReduce (e.g.~\cite{hegeman2015}) and related platforms such as Spark and Hadoop (e.g.~\cite{behnezhad2018}). 

Distance problems, such as single-source shortest path (SSSP) and multi-source shortest path (MSSP), have been widely studied in different models. A fundamental structure that has been used for solving these problems is a hopset. Given a graph $G=(V,E)$, a $(\beta,\epsilon)$-hopset $H$ with hopbound $\beta$, is a set of edges added to $G$ such that for any pair of nodes $u$ and $v$ in $G$, there is a path with at most $\beta$ hops in $G \cup H$ with length within $(1+\epsilon)$ of the shortest path between $u$ and $v$ in $G$. The approximation ratio is also referred to as distortion or \textit{stretch}. We generally want to have sparse hopsets with small hopbound. Intuitively, a hopset can be seen as adding a number of "shortcut edges" that serve as reducing the graph diameter at the expense of a small loss of accuracy. Once a hopset is preprocessed, we can use it as many times as needed for distance queries, and the query time will be the hopbound. 

There is a natural tradeoff between the size and the hop-bound (or the query time) of a hopset. In an extreme-case one could store the complete adjacency list-or equivalently add $O(n^2)$ edges, and then query distance in constant time. Other than the fact that computing all-pairs shortest-path is generally slow, we often do not have enough space to store the whole adjacency list for large-scale graphs.
 There is a line of work that focuses on designing data structures with small size, say $\tilde{O}(n^{1+1/k})$, in which distances can be estimated up to $O(k)$ stretch in small query time. Examples of such structures are Thorup-Zwick distance oracles \cite{thorup2005} or $k$-spanners. Hopsets offer a different tradeoff: a hopset gives more accuracy of $1+\epsilon$ (rather than $O(k)$) at the expense of larger query time (polylogarithmic instead of a small constant). It is therefore crucial to keep the hopbound as small as possible, since the hopbound will basically determine the query time and is more important than preprocessing time. However, even in centralized settings there are existential limitations in this tradeoff. There is a lower bound argument by \cite{abboud2018} stating that there are graphs for which we can not have a hopbound of $o(\log(k)/\epsilon)^{\log(k)}$ and size $O(n^{1+1/k)})$, for arbitrary $0<\epsilon<1$.

In a recent result, Censor-hillel et al.~\cite{censor2019} gave a fast Congested Clique algorithm that constructs a hopset with hopbound $O(\log^2(n)/\epsilon)$ and size $\tilde{O}(n^{3/2})$. While we can use their hopsets to compute distances efficiently, one shortcoming of such a construction is the large space. In particular, if the original graph has size $o(n^{3/2})$, we would be storing more edges than the initial input. This is undesirable due to the large scale nature of data in modern distributed platforms. It is therefore natural to find algorithms that use less space, possibly in exchange for a slightly weaker hopbound (but still polylogarithmic). This is our main goal in this paper.
We extend the result of \cite{censor2019} by constructing sparse hopsets with size $\tilde{O}(n^{1+1/k})$ for a constant $k\geq 2$ and polylogarithmic hopbound in polylogarithmic time in Congested Clique. This is the first Congested Clique construction of sparse hopsets with polylogarithmic hopbound that uses only \textit{polylogarithmic} number of rounds. This implies that we can store a sparse auxiliary data structure that can be used later to query distances (from multiple sources) in polylogarithmic time. 

Our hopset construction is based on a combination of techniques used in Cohen \cite{cohen2000} (with some modifications) and the centralized construction of Huang and Pettie \cite{huang2019}. We also use another result of \cite{censor2019} that allows us to efficiently compute $(1+\epsilon)$-approximate multi-source shortest path distances from $O(\sqrt{n})$ sources.

One tool that we use in our construction is a hop-limited neighborhood cover construction, which may be of independent interest. Roughly speaking, a $W$-neighborhood cover is a collection of clusters, such that there is a cluster that contains the neighborhood of radius $W$ around each node, and such that each node overlaps with at most $O(\log(n))$ clusters. In an $\ell$-limited $W$-neighborhood cover only balls with radius $W$ using paths with at most $\ell$-hops are contained in a cluster.

We note that many of the techniques we use in our construction are borrowed from the PRAM literature. We hope that this paper provides some insight into connections between these different but relevant models.

\subsection{Our contribution} \label{sec:results_intro}

\subparagraph*{Congested Clique.}The state-of-the-art construction of sparse hopsets in Congested Clique is the results of Elkin and Neiman \cite{elkin2019RNC} (and similar results in \cite{elkin2017, elkin2016}), but these algorithms requires polynomial number of rounds for constructing hopsets with polylogarithmic hopbound. The construction of Censor-Hillel et al.~\cite{censor2019} is a special case of hopsets of \cite{elkin2017}. They construct hopsets of size $\tilde{O}(n^{3/2})$ with $O(\log^2(n)/\epsilon)$ hopbounds. They can construct such a hopset in $O(\log^2(n)/\epsilon)$ rounds of Congested Clique using sparse matrix multiplication techniques. However, \cite{censor2019} does not give any explicit results for constructing sparser hopsets. It is possible that their techniques will also lead to faster Congested Clique algorithms for constructing general (sparse) hopsets of \cite{elkin2017, elkin2019RNC}. But here we use a new hopset construction that has a very different structure than hopsets of \cite{elkin2017} and with improved guarantees. Not only does our hopset construction run in polylogarithmic rounds, but it also yields a better size and hopbound tradeoff over the state-of-the-art Congested Clique construction of \cite{elkin2019RNC}. Prior to \cite{elkin2016} and \cite{censor2019} the hopsets proposed for Congested Clique had superpolylogarithmic hopbound of $2^{\tilde{O}(\sqrt{\log(n)})}$ \cite{henzinger2016} or polynomial \cite{nanongkai2014} hopbound.
 More formally, we provide an algorithm with the following guarantees:
\begin{theorem} \label{thm:main}
Given a weighted\footnote{For simplicity, we assume that the weights are polynomial. This assumption can be relaxed using standard reductions that introduce extra polylogarithmic factors in time (and hopbound) (e.g. see \cite{klein1997},\cite{miller2015}, \cite{elkin2016}).} graph $G=(V, E,w)$, for any $k \geq 2, 0< \epsilon \leq 1$, there is a Congested Clique algorithm that computes a $(\beta,\epsilon)$-hopset of size $O(n^{1+\frac{1}{2k}} \log(n) +n \log^2(n))$ with hopbound $\beta={ O(\frac{\log^2(n)}{\epsilon}(\frac{\log(n)\log(k)}{\epsilon})^{\log(k+1)-1})}$ with high probability in $O(\beta \log^2(n))$ rounds. 
\end{theorem}
To comapre this with the \textit{efficient variant} of the Congested Clique hopsets of \cite{elkin2016}, we note that for a hopset of size $O(n^{1+\frac{1}{k}} \log(n))$, for a constant $k$, our hopbound is \\$O(\frac{\log^2(n)}{\epsilon}(\frac{\log(n)\log(k)}{\epsilon})^{\log(k+1)-2})$, whereas \cite{elkin2016} gets a hopbound of $\Omega((\frac{\log(n)\log(k)}{\epsilon})^{\log(k)+2})$. Thus our hopbound guarantee is a factor of $\Theta(\log^{2-o(1)}(n) (\frac{\log(k)}{\epsilon})^{4-o(1)})$ improvement over construction \cite{elkin2016}. Also, their more efficient algorithm runs in $\tilde{O}(n^{\rho})$ rounds of Congested Clique, $0<\rho \leq \frac{1}{2}$ is a parameter that impacts the hopbound ($\rho$ is a constant when their hopbound is polylogarithmic). They have another algorithm that uses an extra polynomial factor in running time to obtain constant hopbound\footnote{If we allow extra polynomial factors in the running time we may also get a constant hopbound (we would need to change the parameters of the neighborhood cover construction, and change how we iteratively use smaller scales). However this is inconsistent with our main motivation of getting polylogarithmic round complexity.}. We note that construction of \cite{elkin2019RNC} has similar guarantees as \cite{elkin2016}, with two differences: it eliminates a $\log(n)$ factor (or more generally the dependence on aspect ratio) in the hopset size, but has a slightly worse hopbound in their fastest regime.

 Our construction is mainly based on ideas of \cite{cohen2000} with few key differences that take advantage of power of Congested Clique. While hopsets of \cite{elkin2016} significantly improve over hopset of \cite{cohen2000} in the centralized settings, construction of \cite{cohen2000} has certain properties that makes it adaptable for a better Congested Clique algorithm. In particular, \cite{cohen2000} uses a notion of small and big clusters, and we can utilize this separation in Congested Clique. We change the algorithm of \cite{cohen2000}, in such a way that leads to adding fewer edges for small clusters. This leads to sparser hopsets and improves the overall round complexity. The key idea is that by using the right parameter settings, in Congested Clique we can send the whole topology of a small cluster to a single node, \textit{the cluster center}, and then compute the best known hopset locally. It is possible to perform these operation specifically in Congested Clique due to a well-known routing algorithm by Lenzen \cite{lenzen2013routing}. We can then combine Theorem \ref{thm:main} with a \textit{source detection} algorithm by \cite{censor2019} (formally stated in Lemma \ref{lem:MSSP}) to get the following result for computing multiple-source shortest path queries.
\begin{corollary} \label{cor:mssp}
Given a weighted graph $G=(V, E,w)$ there is a Congested Clique algorithm that constructs a data structure of size $O(m+n^{1+\frac{1}{2k}})$ in $O(\beta \log^2(n))$ rounds, where $\beta= O(\frac{\log^2(n)}{\epsilon}(\frac{\log(n)\log(k)}{\epsilon})^{\log(k+1)-1})$, such that after construction in $O(\beta)$ rounds we can query $(1+\epsilon)$-stretch distances from $O(\sqrt{n})$ sources to all nodes in $V$ with high probability.
\end{corollary}

In a related result, the problem of single-source shortest path (SSSP) in Congested Clique was also studied by \cite{becker2017}, where they use continuous optimization techniques for solving transshipment. Firstly their algorithm takes a large polylogarithmic round complexity, and has a high dependence on $\epsilon$. But we can have a significantly smaller running time depending on the hopset size. In other words, for hopsets with a reasonable density (e.g. with size $n^{1+\mu}$, where $\mu<0.1$) we get a much smaller polylogarithmic factor for computing $(1+\epsilon)$-SSSP. This can be further reduced if we allow denser hopsets.

More importantly, an approach such as \cite{becker2017} is mainly suited for SSSP. The limitation with their approach is that for computing multiple distance queries we need to repeat the algorithm for each query. For example, for computing the shortest path from $s$ sources to all nodes, we have to repeat the whole algorithm $s$ times. But constructing a hopset will let us run many of such queries in parallel in $O(\beta+s)$ rounds, where $\beta$ is the hopbound. Moreover, we can compute multi-source shortest path from $O(\sqrt{n})$ sources in parallel for all the sources using the source detection algorithm of \cite{censor2019}.

\subparagraph*{Neighborhood and Pairwise Covers.} In Section \ref{sec:covers}, we focus on an efficient construction of a \textit{limited} pairwise cover (or neighborhood cover) in the CONGEST model, which is a tool that we use in our hopset construction.
Given a weighted graph $G=(V,E)$, a $W$-pairwise cover, as defined by \cite{cohen1998}, is a collection $\mathcal{C}$ of subsets of $V$ with the following properties.
1) the diameter of each cluster is $O(W \log n)$,
2) $\sum_{C \in \mathcal{C}} |C| = \tilde{O}(n)$, $\sum_{C \in \mathcal{C}} E(C)=\tilde{O}(m)$. In other words the sum of sizes of all clusters is $O(n)$, and sum of all edge occurrences in the clusters is $\tilde{O}(m)$,
3) for every path $p$ with (weighted) length at most $W$, there exists a cluster $C$ where $p \subseteq C$.   

Pairwise covers are similar to neighborhood covers of Awerbuch and Peleg \cite{awerbuch1990} with two differences: in a $W$-neighborhood cover, there must be a cluster that contains the neighborhood of radius $W$ around each node rather than only paths of length $W$. Neighborhood covers also need an additional property that \textit{each node} is in at most $O(\log(n))$ clusters. While for our purposes the path covering property is enough (rather than neighborhoods), in distributed settings we need the second property that each node overlaps with few clusters to ensure that there is no congestion bottleneck. The main subtlety in constructing a general $W$-pairwise (or neighborhood) cover is that we would have to explore paths of $\Omega(n)$ hops, and thus it is not clear how this can be done in polylogarithmic time. To resolve this, \cite{cohen1998} proposed a relaxed construction called \textit{$\ell$-limited $W$-pairwise cover}. This structure has all the above properties but only for paths with at most $\ell$-hops. More formally, the third property will be relaxed to require that for every path $p$ of weight at most $W$ with at most $\ell$ hops there exists a cluster $C$ where $p \subseteq C$. We can define an $\ell$-limited $W$-neighborhood cover similarly.

A randomized algorithm that constructs $\ell$-limited pairwise covers with high probability in $O(\ell)$ depth in the PRAM model was given by \cite{cohen1998}. The ideas used in \cite{miller2015} for constructing work-efficient PRAM hopsets, can also be used to construct $\ell$-limited pairwise and neighborhood covers in PRAM. However, they do not explicitly construct limited pairwise covers.
In distributed settings, a recent construction for \textit{sparse} neighborhood covers in \textit{unweighted graphs} in the CONGEST model was given by \cite{parter2019}. However, even by generalizing their result to weighted graphs, in order to cover distances for large values of $W$ the algorithm would take $\tilde{\Omega}(W)$ rounds for reasons described above.

To the best of our knowledge, an efficient algorithm for constructing $\ell$-limited pairwise-covers (or limited neighborhood covers) is not directly studied in the CONGEST and Congested Clique literature. Our first contribution is such an algorithm: we use the low-diameter decomposition construction of Miller et al.~\cite{miller2015} for \textit{weighted graphs}, combined with a rounding technique due to \cite{klein1997} to construct $\ell$-limited $W$-pairwise covers in $O(\ell \log^2(n))$ rounds in the CONGEST model.  Importantly, $\ell$ is a parameter independent of $W$, which we will set to a polylogarithmic value throughout our hopset construction. Our algorithm is similar to the algorithm of \cite{miller2015}, but with some adaptations needed for implementation in the CONGEST model. Formally, we get the following result:
\begin{theorem}
Given a weighted graph $G=(V,E,w)$, there is an algorithm that constructs an $\ell$-limited $W$-pairwise cover in $O(\ell \log^2(n) \log(W))$ rounds in the CONGEST model, with high probability. Moreover, a pairwise cover for paths with $\ell$-hops with length in $[W,2W]$ \footnote{The algorithm and analysis can easily be extended to paths with length $[W, cW]$ for any constant $c$.} can be constructed in $O(\ell \log^2(n))$ rounds with high probability.
\end{theorem}

\subparagraph*{MPC.} As a side result\footnote{Our MPC results can be seen as a straight-forward combination of results of \cite{cohen2000}, \cite{miller2015}, \cite{dinitz2019} and simulation of \cite{goodrich2011}. But since both the construction and the model are closely relevant to our Congested Clique algorithms, we find it useful to include this discussion.}, we note that pairwise covers can also be constructed efficiently in the Massively Parallel Computation model (MPC) (even when memory per machine is strictly sublinear). This in turn leads to a better running time for $(1+\epsilon)$-MSSP from $O(\sqrt{n})$ sources (and consequently SSSP), in $O(\log^2(n)/\epsilon)$ rounds in a variant of the model where we assume the overall memory of $\tilde{O}(m\sqrt{n})$ (equivalently, we have \textit{more machines} than in the standard MPC model). We consider this variant since in practice it is plausible that there are more machines, while due to the large-scale nature of data in these settings, using less memory per machine is often more crucial.

Dinitz and Nazari \cite{dinitz2019} construct $(\beta, \epsilon)$-hopsets (based on hopsets of \cite{elkin2016}) with polylogarithmic hopbound when the overall memory is $\tilde{O}(m)$, but they argue that using the existing hopset constructions, this would take polynomial number of rounds in MPC. They further show that if the overall memory is by a polynomial factor larger (i.e. if the overall memory is $\tilde{\Theta}(mn^\rho)$ for a constant $0<\rho \leq 1/2$), then hopsets with polylogarithmic hopbound can be constructed in polylogarithmic time. We can also use this extra-memory idea, first to argue that using hopsets of \cite{cohen2000} instead of hopsets of \cite{elkin2016} we can get a smaller hopbound when the overall memory is $O(m \sqrt{n})$. Then we observe that if we use a faster $\ell$-limited pairwise cover algorithm (based on construction of \cite{miller2015}) instead of pairwise covers that \cite{cohen2000} uses, we can shave off a polylogarithmic factor in the construction time. This $\ell$-limited $W$-pairwise cover construction may also be of independent interest in MPC. More formally, we get faster algorithms for $(1+\epsilon)$-MSSP:
\begin{theorem}
Given an undirected weighted graph $G$, we can compute $(1+\epsilon)$-MSSP from $O(n^{1/2})$ sources in $O(\frac{\log^2(n)}{\gamma\epsilon})$ rounds of MPC, when memory per machine is $\tilde{O}(n^\gamma), 0< \gamma \leq 1$ and the overall memory is $\tilde{O}(mn^{1/2}$) (i.e. there are $\Theta(mn^{1/2 -\gamma})$ machines).
\end{theorem}
 The difference between this result and \cite{dinitz2019} is that they give a more general result where the overall memory is $\tilde{O}(mn^\rho)$ for a parameter $\rho>0$. But in the special case of $\rho=1/2$, we get a hopbound of $O(\frac{\log^2(n)}{\epsilon})$, whereas in this case they get a hopbound of $O((\frac{\log(n)}{\epsilon})^{3})$. %If we have smaller memory (in particular at most $\tilde{O}(mn^{1/3})$), without having a fast MSSP algorithm in MPC similar to the result of \cite{censor2019} in Congested Clique, the constructions based on \cite{elkin2016} lead to slightly better hopbound.
  We also note that the main focus of \cite{dinitz2019} is constructing Thorup-Zwick distance sketches. As explained earlier, these structures offer a different tradeoff: much weaker accuracy ($O(k)$-stretch), but better query time (constant rounds rather than polylogarithmic) and less space after preprocessing ($\tilde{O}(n^{1+1/k})$ instead $\tilde{O}(m)$ in the case of hopsets). More details on the MPC algorithm can be found in Section \ref{sec:MPC}.
 %Finally, we can then use with another idea also proposed by \cite{dinitz2019} to combine this result with a spanner construction so that the algorithm uses $\tilde{O}(m)$ overall memory but the stretch will be $O(k)$.
 
\subsection{Overview of techniques.}
Our hopset has a similar structure to hopsets of \cite{cohen2000}, but with some changes both in construction and the analysis. We also take advantage of multiple primitives that are specific to Congested Clique such as Lenzen's routing and a recent result of \cite{censor2019}. First, we explain the $\ell$-limited $W$-neighborhood cover construction and then we explain the hopset algorithm.

\subparagraph*{$\ell$-limited neighborhood covers.} 
As described earlier, our algorithm for constructing a $W$-neighborhood cover is based on a combination of the low-diameter decomposition of \cite{miller2015}, and a rounding technique originally proposed by \cite{klein1997}. At a high level, in the low-diameter decomposition algorithm of \cite{miller2013, miller2015}, each node $u$ chooses a radius $r_u$ based on an exponential random variable. Then  each node $u$ joins the cluster of  node $v$ that minimizes \textit{the shifted distance} from $u$, which is defined as $d(u,v)-r_v$. This leads to a partition of the graph, and we can show that by repeating this process we will get a $W$-neighborhood cover. Since partitions for constructing a $W$-neighborhood covers directly will be slow for large values of $W$, we focus on the $\ell$-limited $W$-pairwise covers. To construct these, consider all pairs of nodes within distance $[w,2w], w \leq W$ in each iteration. We round up the weights of each edge in the graph based on values $w$ and $\ell$. We then construct a low-diameter-decomposition \textit{based on the new weights}, such that the diameter of each cluster is $O(\ell \log(n))$ (rather than $O(W \log(n))$ based on the original weights). The rounding scheme is such that the $\ell$-limited paths with length $[w,2w]$ in the original graph will be explored.  Intuitively, this means that on the rounded graph we need to explore a neighborhood with fewer hops, which will lead to a faster construction. We can then repeat this process for $O(\log(W))$ times for different distance intervals. The details of this rounding scheme can be found in Section \ref{sec:covers}.

% Therefore we can use a padded decomposition (i.e. partitions that preserve balls of radius $W$), but limited to $\ell$ hops.

 % We assume the aspect ratio, e.g. the ratio between largest and smallest weight, is polynomial. This assumption can be relaxed with extra polylogarithmic overhead in running time (or work/space) using reductions due to Sriram and Klein \cite{klein1997}. Also, w.l.o.g assume that all the edge weights are at least one. Otherwise we can scale up all the edges by the multiplying distances by the smallest weight. 
\subparagraph*{Hopset Construction.} First we describe the sequential hopset construction and will then choose the parameters appropriately for our distributed construction.
 Let $\mu$ be a a parameter that we will set later. The (sequential) structure of the hopset is as follows: In each iteration we consider pair of nodes $u,v \in V$ such that $R \leq d(u,v)< 2R$, and we call the interval $[R,2R)$ a \textit{distance scale}. Then for distance scales $[R,2R)$ we set $W=O(\epsilon R/(\log n))$ and construct a $W$-pairwise cover. We let big clusters be the clusters that have size at least $n^{\mu}$ and small clusters have size less than $n^{\mu}$, where $0< \mu<1$ is a constant parameter. Then we construct a hopset with small hopbound on each of the small clusters. This is the main structural difference with construction of \cite{cohen2000} that adds a clique for the smallest hopsets. We then add a star from the center node of each big cluster to every other node in that cluster, and add a complete graph at the center of large clusters. Whenever we add an edge, we set the weight to be the distance between the two endpoints. In the distributed construction, the weight will be an estimate of this distance that we will describe later.
	
Roughly speaking, constructing a hopset on small clusters rather than constructing a clique as \cite{cohen2000} does, allows us to set the size threshold of small clusters larger, while keeping the number of edges added small. This in turn reduces the number of big cluster centers we have to deal with. Such a modification can be very well tuned to the Congested Clique model. By setting $\mu=1/2$, we will have small clusters that have at most $O(n)$ edges. Then a well-known routing algorithm by Lenzen \cite{lenzen2013routing} can be used to send all these edges to the cluster center. The cluster center can then compute a hopset locally. For this we use current best-known centralized construction by \cite{huang2019}.
The other challenge is that we need to compute pairwise distances between all big cluster centers. In \cite{cohen2000} this step is done by running Bellman-Ford instances from different sources in parallel. But directly implementing this in distributed settings would need $\Omega(\sqrt{n})$ rounds due to congestion. This is where we use a recent result by \cite{censor2019} stating that we can compute $(1+\epsilon)$-approximate distances from $O(\sqrt{n})$ sources in $O(\log^2(n)/\epsilon)$ time. We point out that in \cite{cohen2000}, in order to get sparse hopsets, they use a recursive construction for small clusters. Such a recursion would introduce significant overhead in the hopbound guarantee. Here we show that in Congested Clique by using the tools described above we can avoid using the recursive construction and still compute sparse hopsets.

We explain briefly why the constructed hopset has the size and hopbound properties stated in Theorem \ref{thm:main}. To see this, we use similar arguments as in \cite{cohen2000}: for a distance scale $[R,2R)$ consider a shortest path of length at most $2R$, and consider $O(\log(n)/\epsilon)$ segments of length $W$ on this path. By definition of $W$-pairwise covers, each such segment is contained in a cluster. Therefore we replace this segment with at most $\beta'$ edges for small cluster, where $\beta'$ is the hopbound of the local construction. For big clusters, we either add a single edge, or if there is more than one big cluster, the whole segment between these clusters can be replaced with an edge in the hopset. By similar considerations and by triangle inequality we can show that the stretch of the replaced path is $(1+\epsilon)$. We need a tighter size analysis than the one used in \cite{cohen2000} to prove the desired sparsity. We use a straight-forward bucketing argument as follows: for each cluster of size $\Theta(s)$, $\Theta(s^{1+1/k})$ edges will be added. Then by noting that are at most $O(n/s)$ clusters with this size we can bound the overall size.

\subparagraph*{Bounding the exploration depth.} For large values of $R$, the shortest path explorations up to distance $R$ could take $\Omega(n)$ rounds in distributed settings. To keep the round complexity small, we use the following idea from \cite{cohen2000} (also used in \cite{elkin2016} and \cite{censor2019}): we can use the hopset edges constructed for smaller distance scales for constructing hopset edges for larger distance scales more efficiently. The intuition behind this idea is that any path with length $[R, 2R)$ can be divided into two segments, such that for each of these segments we already have a $(1+\epsilon)$-stretch path with $\beta$ hops using the edges added for smaller distance scales.
This allows us to limit the explorations only to paths with $2\beta+1$ hops in each iteration. This process will impact the accuracy, and so in order to get arbitrary errors, we have to construct the hopsets for a fixed scale at a higher accuracy. This is where a factor polylogarithmic in $n$ will be introduced in the hopbound, which is generally not needed in the centralized constructions (e.g. see \cite{cohen2000, elkin2016}). This idea is formalized in Lemma \ref{lem:rec_hopbound}.

 %This means that for each scale we will lose a $(1+\epsilon)$-factor in stretch, which we will have to rescale later. So instead of a complete graph on the centers of all large clusters, for each distance scale we only add edges between all the large cluster centers that are within $2\beta+1$ hops away from each other. 

\section{Models and Notation}
Given a weighted undirected graph $G=(V,E, w)$, and a pair $u,v \in V$ we denote the (weighted) shortest path distance by $d(u,v)$. We denote by $d^\ell(u,v)$ the length of the shortest path between $u$ and $v$ among the paths that use at most $\ell$ hops, and call this the $\ell$-hop limited distance between $u$ and $v$. For each node $v \in V$, we denote the (weighted) radius $r$ neighborhood around $v$ by $B(v,r)$, and we let $B^\ell(v,r)$ be the set of all nodes $u \in V$ such that there is path $\pi$ of (weighted) length at most $r$ between $u$ and $v$ such that $\pi$ has at most $\ell$ hops. 

For parameter $\beta >0, \epsilon>0$, a graph $G_H = (V, H, w_H)$ is called a $(\beta,\epsilon)$-hopset for the graph $G$, if in graph $G'=(V, E\cup H,w')$ obtained by adding edges of $G_H$, we have $d_G(u, v) \leq d^{\beta}_{G'} (u, v) \leq (1+\epsilon) d_G (u, v)$  for every pair $u, v \in V$ of vertices. The parameter $\beta$ is called the \textit{hopbound} of the hopset.

\subparagraph*{Models.}
We construct limited neighborhood covers in the more classical \textit{CONGEST} model, in which we are given an undirected graph $G=(V, E)$, and in each round nodes can send a message of $O(\log(n))$-bits to each of their neighbors in $G$ (different messages can be sent along different edges). In the \textit{Congested Clique} model, we are given a graph with $n$ nodes, where all nodes can send a message with $O(\log(n))$-bits to \textit{every} other node in the graph in each round \cite{lotker2005}. In other words, this is a stronger variant of the CONGEST model, in all nodes can communicate with each other directly. 

We also consider the \textit{Massively Parallel Computation}, or MPC model \cite{beame2013}. In this model, an input of size $N$ which is arbitrarily distributed over $N/S$ machines, each of which has $S = N^{\gamma}$ memory for some $0 < \gamma < 1$.  In the standard MPC model, every machine can communicate with every other to at most $S$ other machines arbitrarily. Generally, for graph problems the total memory $N$ is $O(m), m=|E|$ words. But here we a consider a variation of the model in which the total memory can be larger, while the memory per machine is strictly sublinear in $n$. In other words, each machine has $O(n^\gamma), \gamma <1$ memory, where $n= |V|$.

Even though we do not give any new PRAM results, we use multiple tools from PRAM literature. In the PRAM model\footnote{We just use a simple abstraction without details of the exact parallel model (EREW, CRCW, etc), since PRAM is not our focus and there are reductions with small overhead between these variants.}, a set of processors perform computations by reading and writing on a shared memory in parallel. The total amount of computation performed by all processors is called the \textit{work}, and the number of parallel rounds is called the \textit{depth}. 
\section{Algorithmic Tools.}
In this section we describe several algorithmic tools from previous work that we will be using. 

\subparagraph*{Bounding the shortest path exploration.} 
As explain earlier, for an efficient hopset construction, we need to first compute hopsets for smaller distance scales and then use the new edges for computing future distances. This will let us limit the shortest path explorations to logarithmic number of hops in each round. More formally,
\begin{lemma}[\cite{cohen2000, elkin2016}]\label{lem:rec_hopbound}
Let $H^k$ be the hopset edges for distance scale $(2^{k-1}, 2^{k}]$ with hopbound $\beta$. Then for any pair $u,v$ where $d(u,v) \in (2^{k}, 2^{k+1}]$, there is a path with $2\beta+1$ hops in $G \cup (\cup_{i=\log(\beta)}^kH^i)$ with length $(1+\epsilon)$-approximate of the shortest path between $u$ and $v$.  
\end{lemma}
Roughly speaking, the above lemma implies that we can use previously added edges and only run Bellman-Ford for $2\beta+1$ rounds for each iteration of our algorithm. %By repeating this process we will incur a multiplicative $(1+\epsilon)$-factor in each iteration, so we have to rescale the error accordingly (which impacts the running time). 
%This allows to construct hopsets for all distance scales using only $\ell=2\beta+1$ limited distances.

\subparagraph*{Lenzen's routing.}
Given a set of messages such that each node is source and destination of at most $O(n)$ messages, these messages can all be routed to their destination in $O(1)$ time in Congested Clique \cite{lenzen2013routing}.

\subparagraph*{Multi-source shortest path and source detection in Congested Clique.} We use the following two results by \cite{censor2019}. First result is a multi-source shortest path algorithm that we use as a subroutine in our hopset construction:
\begin{lemma} [MSSP, \cite{censor2019}]\label{lem:MSSP}
Given a weighted and undirected graph, there is an algorithm that computes $(1+\epsilon)$-approximate distances distances from a set of $O(\sqrt{n})$ sources in $O(\frac{\log^2(n)}{\epsilon})$ rounds in the Congested Clique model.
\end{lemma}
The second result solves a special case of the so-called \textit{source-detection} problem that we use to prove Corollary \ref{cor:mssp}:
\begin{lemma}[Source detection, \cite{censor2019}]
Given a fixed set of $O(\sqrt{n})$ sources $S$, we can compute $\ell$-hop limited distances from all nodes to each of the nodes in $S$ in $O(\ell)$ rounds in the Congested Clique model.
\end{lemma}
\section{Neighborhood covers using low-diameter decomposition} \label{sec:covers}
In this section, we describe an algorithm for constructing pairwise covers in the CONGEST model. We first give an algorithm for $W$-pairwise covers in weighted graphs that runs in $O(W \log^2(n))$ rounds. We then provide an $\ell$-limited $W$-pairwise cover that runs in $O(\ell \log^2(n))$ rounds. Clearly, the CONGEST algorithm can also be used in Congested Clique with the same guarantees. We will use the low-diameter decomposition algorithm that was proposed in \cite{miller2013} and extended (to weighted graphs) in \cite{miller2015} for computing pairwise covers in PRAM. First we state their PRAM result: %Let $0<\alpha<1$ be the padding parameter (i.e. the algorithm chooses from $\exp(\alpha)$).

\begin{theorem}[MPX \cite{miller2015, miller2013}] \label{PRAM_MPX}
Given a weighted an undirected graph $G=(V,E,w)$, there is a randomized parallel algorithm that partitions $V$ into clusters $\mathcal{X}_1,\mathcal{X}_2,...$ such that w.h.p.~the (strong) diameter of each cluster $\mathcal{X}_i$ is at most $O(\frac{\log(n)}{\alpha})$. This algorithm has $O(\alpha^{-1}\log(n))$ depth\footnote{Depending on the exact PRAM model considered the depth may have a small extra factor of $O(\log^*(n))$.} w.h.p.~and $O(m)$ work.
\end{theorem}

We denote the algorithm of \cite{miller2015} for a parameter $0 <\alpha<1$ by LDD($\alpha$), which is as follows: each node $u \in V$ first chooses a random radius $r_u$ based on an exponential distribution $\exp(\alpha)$.  Each node $v \in V$ joins the cluster of node $u= \arg \min_{x \in V} (d(v,x)-d_x)$. Ties can be broken aribtrarily. It is easy to see that based on simple properties of exponential random variables the weak diameter of each cluster is $O(\alpha^{-1}\log(n))$ with high probability. But it can be shown that the clusters also have strong bounded diameter of $O(\alpha^{-1}\log(n))$ (as argued in \cite{miller2013,miller2015}). This means the diameter of the subgraph induced by each cluster is $O(\alpha^{-1} \log(n))$ as opposed to the weak diameter guarantee, which bounds the diameter between each pair of nodes in the cluster based on distaces in $G$. The second property is that we can lower bound the probability that the neighborhood around each node is fully contained inside one cluster by a constant. This was shown in \cite{miller2015}, but we give a proof sketch for completeness.

\begin{lemma}[Padding property\footnote{Lemma 2.2 in \cite{miller2015} upper bounds the probability that a ball overlaps with $k$ or more clusters, but Lemma \ref{MPX_padded} is a straightforward corollary of this claim.}, \cite{miller2015}]  \label{MPX_padded}
 Let $\mathcal{X}$ be a partition in support of the LDD$(\alpha)$ algorithm. For each node $u \in V$, the probability that there exists $C \in \mathcal{X}$ such that $B(u,r) \subseteq C$ is at least $\exp(-2r\alpha)$.
\end{lemma}
\begin{proof}[Proof sketch.] For each node $u$ we will consider the subgraph induced by $B(u,r)$. For each node $v \in V$, consider the random variable $Y_v:= r_v - d(u,v)$. Let $Y_1$ denote the largest $Y_v$ over $v \in B(u,r)$, and let $Y_2$ denote the second largest value. We argue that the probability that $B(u,r)$ intersects more than one cluster is at most $1- \exp(-2r\alpha)$. This event occurs when $Y_1$ and $Y_2$ are within $2r$ of each other. Therefore we only need to bound the probability that $Y_1-Y_2 < 2r$. This now follows from Lemma 4.4. of \cite{miller2013} that claims the following: given a sequence of exponential random variables $r_1, r_2,...,r_n$, and arbitrary values $d_1, d_2,...,d_n$ the probability that largest and second largest values $r_i-d_i$ are within $\delta$ of each other is at most  $1- \exp(-\delta\alpha)$. This implies the probability that $B(u,r)$ intersects more than one cluster is at most $1- \exp(-2r\alpha)$ and this proves the claim. For more details see \cite{miller2013}, \cite{miller2015}.
\end{proof}

In order to compute $W$-neighborhood covers \textit{sequentially}, we can use the above theorem by setting $\alpha=1/W$ and repeating the partition algorithm $O(\log(n))$ times. It follows from a standard Chernoff bound that the desired properties hold with high probability. Implementing this algorithm in distributed settings may take $\Omega(n)$ rounds for large values of $W$. To resolve this, we use a relaxed notion similar to the notion of \textit{$\ell$-limited $W$-pairwise cover} proposed in \cite{cohen2000}. This structure has all the properties of a $W$-pairwise cover but the path covering property only holds for paths with at most $\ell$-hops. More formally, the third property will be relaxed to require that for every path $p$ of weight at most $W$ with at most $\ell$ hops there exists a cluster $C$ where $p \subseteq C$. We define an $\ell$-limited $W$-neighborhood cover similarly: for each node $u$, there is a cluster $C$ such that $B^{\ell}(u,W) \subseteq C$.
In \cite{cohen2000}, Cohen shows that we can construct $\ell$-limited $W$-pairwise covers in $O(\ell)$ parallel depth, independent of $W$. We will show that this concept can also be utilized to limit the number of rounds for LDD($\alpha$) partitions to $O(\ell)$. 

\subparagraph*{$\ell$-limited $W$-pairwise cover.} Since running LDD($\alpha$) by setting $\alpha=1/W$ will require many rounds, we cannot directly use the weighted variant of LDD$(\alpha)$. Instead we use a rounding idea that allows us to run LDD($\alpha$) on the graph obtained from rounded weights, only for $\alpha=O(1/\ell)$, at the cost of a small loss in accuracy. This idea was proposed by \cite{klein1997} and is used widely in PRAM literature (e.g.~\cite{cohen2000}, \cite{miller2015}). In the context of distributed algorithms a similar approach was used by \cite{nanongkai2014} in CONGEST, but directly applying the result of \cite{nanongkai2014} to our settings will require a polynomial running time, since we would need to run the algorithm from many (polynomial) sources. The idea is based on the following observation: consider a path $\pi$ with at most $\ell$ hops, such that $R \leq w(\pi) \leq 2R$ for a fixed $R >0$. Then by slightly changing the weights of each edge $e \in \pi$ by a small additive factor such that for the new weight $\hat{w}$ it holds $w(e) \leq \hat{w}(e) \leq w(e)+ \frac{\epsilon_0 R}{\ell}$ for an arbitrary $\epsilon_0 >0$. We then get $\hat{w}(\pi) \leq w(\pi) +R\epsilon_0 \leq (1+\epsilon_0) w(\pi) $. This can be achieved by setting $\hat{w}(e) = \lceil \frac{w(e)}{\eta} \rceil$, where $\eta=\frac{\epsilon_0 R}{\ell}$. %We get an unweighted graph $\hat{G}$, since w.l.o.g we assumed that edge weights are at least one, which implies $w(e) \leq 2R/\ell, w(e) \leq \eta$. 

\begin{lemma}[\cite{klein1997}]\label{lem:rounding}
Given a weighted graph $G=(V,E,w)$, and a parameter $R$, there is a rounding scheme that constructs another graph $\hat{G}=G=(V,E,\hat{w})$ such that any path $\pi$ with at most $\ell$ hops and weight $R \leq w(\pi) \leq 2R$ in $G$, has $\hat{w}(\pi) \leq \lceil 2\ell/\epsilon_0 \rceil$ in $\hat{G}$. Moreover, $w(\pi) \leq \eta(R,\ell) \cdot \hat{w}(\pi) \leq (1+\epsilon) w(\pi)$, where $\eta(R,\ell) = \epsilon_0 R/\ell$.
\end{lemma}
%\begin{lemma} [$\ell$-limited padding] \label{limited_padding} Let $\mathcal{X}$ be a partition in support of the $\ell$-limited  LDD$(O(1/\ell))$ algorithm. Then for each node $u \in V$, the probability that there exists $C \in \mathcal{X}$ such that $B^{\ell}(u,r) \subseteq C$ is at least $\exp(-2r\alpha)$.\end{lemma}

 We can now run LDD$(\alpha)$ for $\alpha=O(\ell)$ on $\hat{G}$, and each path $\pi$ with at most $\ell$ hops will be fully contained in some cluster with probability at least $\exp(-\ell\cdot O(1/\ell))=\Omega(1)$. We can then recover an estimate to the original length $w(\pi)$ by setting $\tilde{w}(\pi)= \eta(R,\ell) \cdot \hat{w}(\pi)$, and we have $w(\pi) \leq \tilde{w}(\pi) \leq (1+\epsilon_0) \leq w(\pi)$. Same as before, by repeating the LDD$(\alpha)$ algorithm $O(\log(n))$ times we will get an $\ell$-limited $W$-neighborhood cover.
We first argue that this algorithm can be implemented $O(\alpha \log^2(n))$ rounds of the CONGEST model. A similar construction was used in \cite{parter2019} for $W$-neighborhood covers in CONGEST. But the result of \cite{parter2019} only focuses on unweighted graph, and would take $O(W \log(n))$ rounds. %Therefore, we extend the result of \cite{parter2019} to weighted graphs and to $\ell$-limited neighborhood covers.

\begin{theorem}\label{thm:RoundCovers}
Given a weighted graph $G=(V,E,w)$, there is an algorithm that constructs an $\ell$-limited $W$-pairwise cover in $O(\ell \log^2(n) \log(W))$ rounds in the CONGEST model, with high probability. Moreover, a pairwise cover for paths with $\ell$-hops with length in $[W,2W]$ can be constructed in $O(\ell \log^2(n))$ rounds with high probability.
\end{theorem}
\begin{proof} 
As we argued by using the rounding technique of \cite{klein1997}, for any pair of nodes $u,v$ such that $d^{\ell}(u,v) \in [W, 2W)$ we can restrict our attention to another graph $\hat{G}$ with rounded weights. We construct a pairwise cover on $\hat{G}$ by running the LDD$(\alpha)$, $\alpha= \epsilon_0/2\ell=\Theta(1/\ell)$ algorithm $O(\log(n))$ times independently.

 We argue that each run of LDD$(O(1/\ell))$ takes $O(\ell)$ rounds in the CONGEST model. First we observe that each node $u$ only needs to broadcasts the value $r_u$ to all the nodes within its $r_u$ neighborhood, since  a node $x$ will not join the cluster of $u$ if $r_u-d(u,x) <0$. We can now use a simple induction to prove the claim. In each round, each node $u$ will forward the radius and distances corresponding to the node $u_{\max}$ that maximizes $r_{u_{\max}} -d(u,u_{\max})$ among over all the messages that $u$ has received. We now argue that each node $u$ will receive the message from the node $c =\arg \max_{v \in V} r_v -d(u,v)$ in $r_c$ rounds. Consider any path $\pi =\{c=u_0,u_1,...,u_{j}=u\}$, where $j \leq r_c$. If $j=1$ then in one single round $c$ sends $(r_c,w(c,u_1))$ to $u_1$. Assume now that $u_i$ receives the message $(r_c, d(u,u_{i-1}))$ in round $i$. Then $u_i$ will compute $d(u,u_i)$ (after receiving distance estimates from all neighbors), and forwards $(r_c,d(c,u_i))$ to all neighbors including $u_{i+1}$. Therefore in round $i=j \leq r_c$, $u_j$ has received the message $r_c-d(u,c)$, and can compute $d(u,c)$. Therefore this algorithm will terminate 
after $\max_{v \in V} r_v$ rounds. Since $r_u$ is an exponential random variable with parameter $O(1/\ell)$, we know that maximum of these $O(n)$ exponential random variables is $O(\ell \log(n))$ with high probability.
Now we need to repeat the partition algorithm $O(\log(n))$ times and will pipeline the broadcasts for different runs. Clearly, each node is in at most in $O(\log(n))$ clusters. A standard Chernoff bound in combination with Lemma \ref{MPX_padded} implies that with high probability after $O(\log(n))$ repetition of the $\ell$-limited LDD$(O(1/\ell))$ algorithm, for each path $\pi$ with at most $\ell$ hops and length $w(\pi) \in [R, 2R]$, there will be a cluster $C$ such that $\pi \subseteq C$.
We then repeat this process for $O(\log(W))$ distance scales to get an $\ell$-limited $W$-pairwise cover.
\end{proof}
As we will see, since in our hopset construction we consider different distance scales and need to compute pairwise covers for a fixed scale, this step takes only $O(\ell \log(n))$ rounds.

\subparagraph*{Diameter guarantee.} For constructing pairwise covers, we need the diameter guarantee of $O(W \log(n))$ for all clusters. While running LDD$(O(1/\ell))$ gives a diameter guarantee of $O(\ell \log(n))$ on $\hat{G}$, we note that the construction ensures that clusters have diameter $O(W \log(n))$ on $G$. Since we argued that every $\ell$ hop path with length $W$ will fall into a cluster with high probability, the diameter guarantee of $O(\ell \log(n))$ on $\hat{G}$ will imply that the corresponding cluster in $G$ will have length $O(W \log(n))$. More formally, for any pair of nodes there is a path with length $O(\ell \log(n))$ in $\hat{G}$. Let $\mathcal{C}$ denote the cluster that contains this path. Consider each segment of length $O(\ell)$ in $\hat{G}$ is in $\mathcal{C}$ and will be have length $O(W)$ in $G$ (by Lemma \ref{lem:rounding}), and thus there will be a path of length $O(W\log(n))$ based on weights in $G$ in $\mathcal{C}$.

\subparagraph*{Extension to neighborhood covers.}
While Cohen shows that for the parallel construction of hopsets pairwise covers are enough, for distributed implementation we need one more property: each vertex should overlap with at most $O(\log(n))$ clusters. Moreover, the algorithm used in Theorem \ref{thm:RoundCovers} provides the stronger guarantee that there will be a cluster that contains the neighborhood of (weighted) radius $W$ from each vertex with high probability, rather than only containing paths of length $W$. 
In other words, a similar analysis shows that with high probability an $\ell$-limited neighborhood cover can be constructed in $O(\ell \log(n))$ rounds of the CONGEST model. That is, for each node $u$, the $\ell$-limited $W$-neighborhood of $u$ will be fully contained in a cluster with high probability. However, for our purposes the path covering property suffices.

\iffalse More formally, a $W$-neighborhood cover consists of a set of clusters $\mathcal{X}$ such that: for each node $u$ there is a cluster $\mathcal{C} \in \mathcal{X}$ that contains $B(u,W)$ (i.e. $W$ neighborhood of $u$), each cluster has diameter $O(W \log(n))$ and each vertex belongs to at most $O(\log(n))$ clusters in $\mathcal{X}$. Parter and Yogev show that we can construct neighborhood covers for unweighted graphs in $O(W \log(n))$ rounds in the CONGEST model. Even if disregard weights, we need to use these covers for large values of $W$, such as $W=\Omega(n)$. To resolve this issue we use a notion of \textit{$\ell$-limited $W$-neighborhood cover}, in which for each node $u$ only nodes in $W$-neighborhood of $u$ using paths with at most $\ell$ hops need be contained in a cluster. In other words, there must be a cluster $C \in \mathcal{X}$ such that $B^{\ell}(u,W) \subseteq X$ for each node $u \in V$.

It is not hard to see that analysis we used can be extended to show that with high probability an $\ell$-limited neighborhood cover can be constructed in $O(\ell \log(n))$ rounds of the CONGEST model. That is, for each node $u$, the $\ell$-limited $W$-neighborhood of $u$ will be fully contained in a cluster with high probability. However for our purposes the path covering property suffices.
\fi
\section{Congested Clique Hopset Construction.}
In this section we describe our main algorithm. Similar to the sequential construction described we consider different distance scales $[R,2R)$, and handle each scale separately. In each iteration, we construct a \textit{sparse} $\ell$-limited $2R$-neighborhood cover as described in Section {\ref{sec:covers}}. Then the clusters will be divided into small and big clusters, and each case will be handled differently. So far the construction is similar to \cite{cohen2000}.
The key new idea is that for Congested Clique, by setting the parameters carefully we can send the topology corresponding to a small cluster to the cluster center, and build a hopset locally. Here we need to use the fact that each node is in at most $O(\log(n))$ clusters, which is a property that we get from our neighborhood cover construction. We will also need to compute pairwise distances between big clusters centers. For this step, we use the algorithm of \cite{censor2019} that computes $(1+\epsilon)$-multi-source shortest path from $O(\sqrt{n})$ sources (Lemma \ref{lem:MSSP}). We note that while during our construction we construct the denser hopsets of \cite{censor2019} as auxiliary structure, these extra edges will be removed at the end of each distance scale.

Finally, we use Lemma \ref{lem:rec_hopbound} to use the hopset edges added for smaller distance scales to construct the larger distance scales. For this to give us a $(1+\epsilon)$ for an arbitrary $\epsilon$, we first let $\epsilon'$ be the error parameter. Since we use paths with error $(1+\epsilon')$ for each scale, to compute distances for the next scale, a multiplicative factor in the stretch will be added in each iteration.
This means that after $i$ iterations the error will be $(1+\epsilon')^i$. We can simply rescale the error parameter by setting $\epsilon'=O(\frac{\epsilon}{\log(n)})$ to get arbitrary error overall of $\epsilon>0$. 

Throughout our analysis w.l.o.g we assume the minimum edge weight is one. Otherwise, we can scale all the edge weights by the minimum edge weight. We also assume the aspect ratio is polynomial. Otherwise we can use reductions in previous work to reduce aspect ratio in exchange in polylogarithmic depth (this will be preprocessing step and will not dominate the overall running time).

\begin{algorithm}[h]
\caption{Congested Clique construction $(\epsilon,\beta)$-hopset of size $\tilde{O}(n^{1+\frac{1}{2k}})$}
\label{alg:main}
Let $H_i$ denote the hopset edges for scale $(2^i, 2^{i+1}]$, and set $\epsilon'=O(\frac{\epsilon}{\log(n)})$.\\
\For{$(R, 2R]$, where $R=2^\kappa, \log(\beta)\leq \kappa \leq O(\log(n))$, on $G \cup^{\kappa-1}_{i \geq \log(\beta)} H_i$}{
  Set $W=O(\epsilon' R/(\log n))$, and and build $\beta$-limited $W$-pairwise covers (by Theorem \ref{thm:RoundCovers}).\\
  Let $\mathcal{C}_b$ be the set of big clusters that have size at least $\sqrt{n}$ and $\mathcal{C}_s$ small clusters to have size less than $\sqrt{n}$.\\
\For{each $C \in \mathcal{C}_s$}{
All the nodes in cluster $C$ send their incident edges to the center.\\
The center locally computes a $(\epsilon',\beta')$-hopset of size $O(n^{1/2+1/2k})$ (construction of \cite{huang2019}) with $\beta'= O(\log(k)/\epsilon')^{\log(k+1)-1}$.\\
The center sends the new hopset edges to the corresponding nodes (endpoints) in $C$.
}
\For{each $C \in \mathcal{C}_b$}{
Add a star by adding edges from the center node of each big cluster to every other node in that cluster (limiting exploration to $\ell=2\beta+1$ hops), and set the weights based on shortest path distances.
}
Add an edge between any pair $u_1, u_2$ of centers of big clusters that are within $\ell=2\beta+1$ hops of each other, and set the weight to $d^{\ell}(u_1, u_2)$.\\
}
\end{algorithm}

An overview of the algorithm is presented in Algorithm \ref{alg:main}. By defining small clusters to have size at most $\sqrt{n}$, we have that the number of edges in each small cluster $C$ is $O(n)$, and hence all the nodes $C$ can send their incident edges to the cluster center in constant rounds using Lenzen's routing \cite{lenzen2013routing}. Then the cluster center computes a hopset with size $O(n^{1/2+1/2k})$ and hopbound $\beta_0= O(\log(k)/\epsilon')^{\log(k)-1}$ \textit{locally} using Huang-Pettie \cite{huang2019} centralized construction.  The center of a small cluster $C$ can send the edges incident to each node in that clusters. Since the size of the hopset on small clusters is always $O(n^{\frac{1}{2}+\frac{1}{2k}})=O(n)$, this can also be done in constant  time using Lenzen's routing.

 As explained in Lemma \ref{lem:rec_hopbound}, using hopset edges added for smaller scales we can limit all the shortest path explorations to $2\beta+1$. So we can add the star edges, by running $2\beta+1$ rounds of Bellman-Ford.
For adding a clique between centers of large clusters, we will use the $(1+\epsilon)$-MSSP algorithm of Censor-Hillel et al.~2019 \cite{censor2019} (using Lemma \ref{lem:MSSP}. This is possible since there are at most $O(\sqrt{n})$ nodes. We disregard all the other edges added in this step for computing these distances after the computation.
We now analyze the algorithm and show that it has the properties stated in Theorem \ref{thm:main}.
\subparagraph*{Hopbound.}  
The hopbound of hopsets constructed for small clusters, which are based on Huang-Pettie hopsets is $\beta'=O((\log(k)/\epsilon')^{\log(k+1)-1})$. The path between a pair of nodes in each distance scale has $O(\log(n)/\epsilon')$ segments. The properties of a neighborhood cover imply that each of these segments are w.h.p. contained in one cluster. Each such segment has a corresponding path with hopbound either $\beta'$ (if it is contained in a small clusters), or two edge corresponding to the stars edges (if there is only one big cluster). 
If there are more than one big cluster centers, there is a single edge between the furthest big cluster centers on the path.
Hence in the worst case all segments correspond to small clusters and will have a corresponding path of length $\beta'$. Therefore the overall hopbound is $O(\frac{\log(n)}{\epsilon'}(\frac{\log(k)}{\epsilon'})^{\log(k+1)-1}))=O(\frac{\log^2(n)}{\epsilon}(\frac{\log(n)\log(k)}{\epsilon})^{\log(k+1)-1})$.

\subparagraph*{Size.}
Recall that large clusters have size at least $\sqrt{n}$. The stars added for each big cluster will add $O(n \log^2(n))$ edges overall since they are consisted of unions of $O(\log(n))$ forests for each scale. The (clique) edges added between centers of big clusters will add $O(n)$ edges overall. For small clusters of size $s = O(\sqrt{n})$, we added a hopset of size $s^{1+1/k}$ (this is the guarantee we get by using Huang-Pettie \cite{huang2019} hopsets), for a parameter $k \geq 2$. On the other hand, we have at most $O(\frac{n}{s})$ clusters of size within $[s, 2s]$. Therefore we can estimate the overall number of edges added for these small clusters in each scale by summing over different values $s \in [2^r, 2^{r+1})$ for small clusters as follows:

\[\sum_{s \in [\sqrt{n}]} O(\frac{n}{s}\cdot s^{1+1/k})=\sum_{r=1}^{\log (\sqrt{n})} O(\frac{n}{2^r} \cdot (2^r)^{1+1/k}) = \sum_{r=1}^{\log (\sqrt{n})} O(n \cdot 2^{\frac{r}{k}})= O(n^{1+\frac{1}{2k}}). \]

Therefore, the overall size for all scales is $O(n^{1+\frac{1}{2k}} \log(n)+n\log^2(n))$. 

%recursive:For small clusters of size $s \in [n^\mu, n^{1-\rho}]$ we have use an inductive argument. Assume that at the $i-1$-th level of the reduction we have added a hopset of size $O(s^{1+\eta}), \eta <\frac{1}{2}$ for each cluster of size $s$. Also, at level $i$ we have $s=n^{i(1-\rho)}$. There are at most $n/s$ clusters of size $s$ and each of these clusters may contain $\frac{n^{i(1-\rho)}}{n^{(i-1)(1-\rho)}}= n^{1-\rho}$ clusters of level $i-1$.
\subparagraph*{Stretch.} 
 Fix a distance scale $(R, 2R], R=2^k$ and consider a pair of nodes  $u,v \in V$ where $d(u,v) \in (R,2R]$.  If $R \leq \log \beta$, since we assumed the minimum edge weight is one, this implies that the shortest path has at most $O(\beta)$ hops and no more edges is needed for this pair.
Otherwise, let $\pi$ be the shortest path between $u$ and $v$. We look at three different cases and show $\tilde{d}(u,v)= (1+\epsilon') d(u,v) + O(W \log n)$.

First consider the case where all the clusters on the shortest path between $u$ and $v$ are small clusters. In this case, we have replaced each segment of length $W$ with a path of stretch $(1+\epsilon')$. By triangle inequality, overall we get a $(1+\epsilon')$-stretch.
  Next, consider a case where there is a single large cluster on this path. The segment corresponding to this single cluster will just add a single additive $W \log n = \epsilon' R$ cost to our distance estimate.

Final case is when there are more than one large clusters. Consider the two furthest large clusters (based on their centers) on $\pi$, and let their centers be $x$ and $y$. We have added one single edge within $(1+\epsilon')$-stretch of $d(x,y)$ that covers the whole segment between these two centers. Therefore, we have shown that all segments of $\pi$ have a corresponding path within $(1+\epsilon')$-stretch. As argued, this implies that each scale incurs a multiplicative factor of $(1+\epsilon')$ in the stretch, and thus by setting $\epsilon'=O(\epsilon/\log(n))$, and since we assumed that the weights are polynomial we can get $(1+\epsilon)$-stretch for all scales.

\subparagraph*{Round Complexity.} For each of the $O(\log(n))$ distance scales, it takes $O(\beta \log^2(n))$ rounds to compute $2\beta+1$-limited neighborhood covers (Lemma \ref{thm:RoundCovers}). Once the covers are constructed for small clusters we need to run a Bellman-Ford with $O(\beta)$ hops from the center of each big cluster and since each node may overlap with at most $O(\log(n))$ clusters this phase takes $O(\beta \log(n))$ (each node can pipeline the computation over the clusters it overlaps with). For small clusters, we argued that in $O(1)$ rounds (using Lenzen's routing) the whole small cluster topology can be sent to the cluster center, and after local computation another $O(1)$ rounds will be enough for cluster center to send back the new hopset edges to the destination node. Finally, using the result of \cite{censor2019} we can compute $(1+\epsilon')$-approximation from big cluster centers ($O(\sqrt{n})$ sources) in $O(\frac{\log^2(n)}{\epsilon'})= O(\frac{\log^3(n)}{\epsilon})$ time. Therefore the overall running time is $O(\beta \log^2(n))$.

\subparagraph*{Application to multi-source queries.} We can now combine our hopset construction with Lemma \ref{lem:MSSP} (source detection algorithm of \cite{censor2019}) to show that we can compute queries from $O(\sqrt{n})$ sources in $O(\beta)$ time, by maintaining a sparse hopset of size $\tilde{O}(n^{1+ \frac{1}{2k}})$, while \cite{censor2019} has to store a hopset of size $\tilde{O}(n^{3/2})$. Corollary \ref{cor:mssp} follows from this observation.

\section{Massively Parallel Hopsets and MSSP} \label{sec:MPC}
In this section, we argue that in a variation of the MPC model where the overall memory is $O(mn^{1/2})$ we can construct hopsets with small hopbound efficiently, and this in turn gives us fast algorithm for multi-source shortest path in this case. This result relies on an observation made in \cite{dinitz2019}, stating that the PRAM hopset constructions (e.g. \cite{cohen2000}, \cite{elkin2016}) that use $O(m\alpha)$ processors with depth $t$ can be implemented in MPC, even when the memory per machine is strictly sublinear, in $O(t)$ rounds if we assume that the overall memory available is $O(m\alpha)$. Once a $(\beta, \epsilon)$-hopset is constructed, the Bellman-Ford subroutine described in \cite{dinitz2019} can be used to compute $(1+\epsilon)$-stretch distances from $O(\sqrt{n})$ nodes to all other nodes in $V$.

Results of \cite{dinitz2019} are based on hopsets of \cite{elkin2016}, and their constructions may use less overall memory in general, but they get a worse hopbound than ours in the special case that the total memory is $\tilde{\Omega}(m\sqrt{n})$. In this case, we get an improved hopbound of $O(\log^2(n)/\epsilon)$, whereas their result gives a hopbound of $O((\log(n)/\epsilon)^{3})$.
 In particular, we use the PRAM hopset construction of \cite{cohen2000} (instead of \cite{elkin2016}), which can be simulated in the MPC model with strictly sublinear memory per machine (using a reduction of \cite{goodrich2011}) to construct hopsets with hopbound $O(\log^2(n)/\epsilon)$. The only difference between our construction and \cite{cohen2000} is using a faster algorithm for constructing $\ell$-limited pairwise covers based on the algorithm of \cite{miller2015}. 
 First we note that our $\ell$-limited $W$-neighrbohood cover construction can be constructed in MPC based on a very similar algorithm and analysis as in Section \ref{sec:covers}. This step can be done only using $O(m\log^2(n))$ overall memory (or $O(m \log(n))$ memory for a single-scale) in $O(\ell \log(n))$ rounds. Observe that the construction of $W$-neighborhood covers for different scales $[W, 2W]$ can all be done in parallel with an extra logarithmic overhead in the \textit{total memory}. Similarly since each of the low diameter partitions are independent, the repetitions of the LDD algorithm can also be parallelized. This result is also implied by results in \cite{miller2015}, combined with goodrich \cite{goodrich2011}. We have,
\begin{lemma}\label{lem:MPC_covers}
There is an algorithm that runs in $O(\frac{\ell}{\gamma}\cdot \log(n))$ rounds of MPC and w.h.p.~computes an $\ell$-limited $W$-neighborhood cover, where memory per machines is $O(n^{\gamma}), 0<\gamma \leq 1$ and the overall memory is $O(m \log^2(n))$.
\end{lemma} 
\subparagraph*{$(1+\epsilon)$-MSSP.}
Given a pairwise cover, assuming that in MPC we have $\tilde{O}(mn^{1/2})$ total memory, we can construct a $(\log^2(n)/\epsilon, \epsilon )$-hopset of size $O(n^{3/2} \log(n))$.
 This hopset is a special case of hopsets of \cite{cohen2000}: we add a clique for small clusters, a star centerd at each big cluster, and a clique between big cluster centers. As stated, the main difference in our algorithm is that we use the algorithm of Lemma \ref{lem:MPC_covers} for constructing pairwise covers, rather than the algorithm of \cite{cohen1998}. This leads to a construction time of $O(\beta \log(n))$, whereas a direct reduction from \cite{cohen2000} would have construction time of $O(\beta \log^3(n))$, which is how long it takes to construct their limited pairwise covers. Hence combining Lemma \ref{lem:MPC_covers} with simulating the PRAM construction of \cite{cohen2000}, and the Bellman-Ford primitives described in \cite{dinitz2019}, we can construct a hopset of size $O(n^{3/2} \log(n))$ in $O(\beta \log(n))$ time with hopbound $\beta=O(\log^2(n)/\epsilon)$.
\begin{theorem}
Given an undirected weighted graph $G$, and parameters $\epsilon>0, 0 < \gamma \leq 1$, we can w.h.p.~construct an $(\beta, \epsilon)$-hopset of size $O(n^{3/2}\log(n))$ in  $O(\frac{\log^3(n)}{\gamma\epsilon})$ rounds of MPC, using $O(n^\gamma)$ memory per machine, and the overall memory of $O(mn^{1/2})$ (i.e. there are $O(mn^{1/2-\gamma})$ machines), where hopbound is $\beta=O(\frac{\log^2(n)}{\gamma\epsilon})$.
\end{theorem}
The analysis is very similar to the arguments in previous sections and previous work. Similarly, for $(1+\epsilon)$-MSSP we get, 
\begin{theorem}
Given an undirected weighted graph $G$, after a preprocessing step of $O(\frac{\log^3(n)}{\gamma\epsilon})$ rounds, we can w.h.p.~compute $(1+\epsilon)$-multi source shortest path queries from $O(n^{1/2})$ sources in $O(\frac{\log^2(n)}{\gamma\epsilon})$ rounds of MPC, when the memory per machine is $O(n^\gamma), 0< \gamma \leq 1$, and the overall memory required for preprocessing is $O(mn^{1/2})$.
\end{theorem}
At a high-level since we have overall memory of $\tilde{O}(mn^{1/2})$, to each node $u$, we can assign a block of memory of size $O(\deg(u). n^{1/2})$. Then using aggregations primitives (e.g. see \cite{dinitz2019}), we can store and update the distances from up to $O(n^{1/2})$ sources. Therefore given a hopset with hopbound $O(\log^2(n)/\epsilon)$, we can compute distances from $O(n^{1/2})$ sources by running parallel Bellman-Ford. 
\small{\bibliography{CongestedClique_Hopsets}}
%\appendix

\end{document}